\documentclass[journal,onecolumn,]{IEEEtran}
\usepackage{amsmath} 
\usepackage{amssymb}  
\usepackage{color}
\usepackage{dsfont}
\usepackage{setspace}
\newtheorem{theorem}{Theorem}
\newtheorem{lemma}[theorem]{Lemma}

\newtheorem{definition}{Definition}

\newtheorem{assumption}{Assumption}
\newtheorem{proof}{Proof}

\begin{document}

\title{
A Risk Aware Two-Stage Market Mechanism for Electricity with Renewable Generation
}%

\author{Nathan Dahlin and Rahul Jain
\thanks{The authors are with the ECE Department at the 
        University of Southern California. Email:
        {\tt\small \{dahlin,rahul.jain\}@usc.edu}. This work was supported by NSF Award ECCS-1611574. }%
}

\maketitle
\vspace{-0.5in}
\begin{abstract}

Over the last few decades, electricity markets around the world have adopted multi-settlement structures, allowing for balancing of supply and demand as more accurate forecast information becomes available. Given increasing uncertainty due to adoption of renewables, more recent market design work has focused on optimization of expectation of some quantity, e.g. social welfare. However, social planners and policy makers are often \textit{risk averse}, so that such \textit{risk neutral} formulations do not adequately reflect prevailing attitudes towards risk, nor explain the decisions that follow. Hence we incorporate the commonly used risk measure \textit{conditional value at risk} (CVaR) into the central planning objective, and study how a two-stage market operates when the individual generators are risk neutral. Our primary result is to show existence (by construction) of a sequential competitive equilibrium (SCEq) in this risk-aware two-stage market. Given equilibrium prices, we design a market mechanism which achieves social cost minimization assuming that agents are non strategic. 

\end{abstract}

\section{Introduction}

Electricity markets covering the majority of the US, and most of the industrialized world are operated as multi-settlement markets. These markets are organized in the sense that demand for and supply of energy and ancillary services are matched via a centralized auction mechanism, as opposed to bilateral negotiations over individual transactions \cite{GERARD201819}, \cite{DBLP:conf/hicss/MoyeM18}. An independent system operator (ISO) runs a given market as a series of multi stage forward markets, and a real time or spot market. Depending upon the region, forward markets are settled days or hours ahead of the intended time of delivery, allowing for provision of cheaper, but relatively inflexible generation. Spot markets open and are settled minutes before delivery in order to balance supply and demand in real time. 

While the clearing mechanisms in electricity markets are designed with the objective of maximizing welfare of both producers and consumers, the imperative to increase penetration of renewables and reduce reliance on fossil fuels now strains the existing multi-settlement paradigm. In the past, the primary source of uncertainty in market clearing came from demand side deviations, and such errors reduced to under 5\% by the opening of spot markets \cite{forecasterror}. Current levels of renewable adoption have already exacerbated typical levels of uncertainty for net demand (demand minus renewable generation), and under official mandates to increase the proportion of energy sourced from renewables this trend will only continue. For example, in California renewables constitute nearly 34\% of total retail sales, while a recently passed state bill legislates that 100\% of power come from renewables by 2045 \cite{CECreportJan2019}, \cite{cal100percent}. Reliance on renewables to this degree introduces an order of magnitude greater uncertainty in net demand, and necessitates novel market designs to address this challenge. 

As a starting point, given that economic dispatch is a multi-settlement process, it makes sense to couple markets across forward and real time stages, and then allow for recourse decisions, rather than settle each stage independently. If a probability distribution for the stochastic generation is known, then maximization of \textit{expected} welfare is a reasonable objective. This type of problem can be formulated as a two-stage stochastic program, and in fact it is possible to show that stochastic clearing is more efficient than two-settlement systems \cite{GERARD201819}, \cite{morales2013integrating}. 

There are a couple of issues with the use of expected social welfare as an objective function. In purely mathematical terms, a given realization of a random variable can be quite different from its expectation. Thus optimization of an expectation guarantees little in terms of variation over possible outcomes. Further, real-world observations indicate that economic decision makers are risk averse, or at least act so \cite{bublitz2019survey}. Therefore, given increasing levels of generation variability, it is of both theoretical and practical interest to incorporate some notion of risk into market objective functions. 

In this paper we study how the introduction of risk preferences into the central objective function affects market operation. We consider a setting with an ISO and multiple generators. The ISO owns a nondispatchable, renewable resource, and the market clears in two stages: a forward stage in which only a forecast for renewable generation is available, and a real time stage, wherein the exact realized renewable generation is known. The generators each own primary and ancillary plants, which may be dispatched in the forward and real time stages, respectively. In the forward market, the ISO schedules primary energy procurements from the generators, and in the real time market purchases ancillary service where necessary. All participants are assumed to be non-strategic \textit{price takers}. However, while we assume that the generators seek to maximize their expected profit, the ISO is risk averse and minimizes a weighted sum of the expectation and \textit{conditional value at risk} (CVaR) of its costs. CVaR has over the past two decades become the most widely used risk measure, due to the fact that it is a \textit{coherent} risk measure, and can be calculated via a convex program \cite{de2017investment}. 

Our main result is the proof of existence of a sequential competitive equilibrium (SCEq) in this risk aware, two-stage market with recourse. In particular, we demonstrate the existence of first and second stage prices such that, given these prices, the generation decisions of the generators in both decisions achieve market clearance in stage two, thus balancing supply and demand. We then specify a two-stage market mechanism which implements the SCEq.

\textbf{Related work.} Numerous past works have studied market and mechanism design and equilibrium outcomes in the two-stage expected welfare maximization, or \textit{risk neutral} setting, e.g., \cite{dahlin2019twostage}, \cite{wannegkowshameysha11b} and \cite{xu2015efficient}.

Turning to literature which incorporates risk preferences, several works consider settings in which agents may enter into contracts in order to hedge against risky outcomes. In \cite{ralph2015risk} it is shown that a \textit{complete} market, wherein all uncertainties can be addressed via a balanced set of contracts, involving agents equipped with coherent risk measures, is equivalent to one in which said agents are risk neutral, and take actions based on a probability density function determined by a system risk agent. The work then investigates necessary and sufficient conditions for existence of an equilibrium consisting of allocations, prices and contracts. Assuming a similar setting in the context of hydro thermal markets, \cite{philpott2016equilibrium} shows that given a sufficiently rich set of securities are available to risk averse agents, that a multi-stage competitive equilibrium may be derived from the solution to a risk-averse social planning solution. \cite{GERARD201819} investigates difficulties that may arise when risk averse agents maximize their welfare in a market are not complete, including existence of multiple, potentially stable equilibria. Our setting differs from these works in that we have one risk aware customer for multiple risk neutral producers, and we do not allow for transactions between agents outside of the quantities of energy purchased and consumed. 

\section{Preliminaries}

\subsection{Risk Measures}
In stochastic optimization we are concerned with losses $Z(\omega) = L(x,\omega)$ that are both a function of a decision $x$, as well as some random outcome $\omega$, unknown when the decision is made. Generally speaking, a risk measure is a functional which accepts as input the entire collection of realizations $Z(\omega)$, $w\in\Omega$. 

More specifically, consider a sample space $(\Omega,\mathcal{F})$ equipped with sigma algebra $\mathcal{F}$, on which random functions $Z=Z(\omega)$ are defined. A \textit{risk measure} $\rho(Z)$ maps such random functions into the extended real line \cite{ShapDentRusz09}. Often times the domain of $\rho$, denoted $\mathcal{Z}$ is taken as $\mathcal{L}_p(\Omega,\mathcal{F},P)$ for some $p\in[1,+\infty)$ and reference probability measure $P$. The following characteristics of risk measures will become useful in later sections. 

\begin{definition}
A \textit{proper} risk measure satisfies $\rho(Z)>-\infty$ for all $Z\in \mathcal{Z}$ and $$\text{dom}(\rho):=\{Z\in\mathcal{Z}\,:\,\rho(Z)<\infty\}.$$
\end{definition}
We denote by $Z\succeq Z'$ the pointwise partial order, meaning $Z(\omega)\geq Z'(\omega)$ for a.e. $\omega\in \Omega$. 
\begin{definition}
A risk measure is \textit{monotonic} if $Z,Z'\in\mathcal{Z}$ and $Z\succeq Z'$ implies $\rho(Z)\geq\rho(Z')$.\end{definition}
\begin{definition}
A risk measure is \textit{coherent} if it is monotonic, convex and satisfies translation equivariance and positive homogeneity (see \cite{ShapDentRusz09} for details on these properties).\end{definition}

\subsection{Conditional value at risk}
In the following sections, we will focus in particular on conditional value at risk, or CVaR. CVaR is an example of a coherent risk measure \cite{ShapDentRusz09}. Before defining CVaR, we introduce the related quantity, value at risk. 

Suppose that random variable $Z$ is distributed according to Borel probability measure $P$, with associated sample space $(\Omega,\mathcal{F})$, and cdf $F$. When $Z$ represents losses, the $\alpha\text{-\textit{Value-at-Risk}}$ is defined as follows.
\begin{definition}\label{VaRdef}
For a given confidence level $\alpha\in(0,1)$, the $\alpha\textit{-Value-at-Risk}$ or  $\text{VaR}_{\alpha}$ of random loss $Z=Z(\omega)$ is  
\begin{equation}
\text{\normalfont{VaR}}_{\alpha}(Z) = \min\{z\,:\,F(z)\geq \alpha\}.
\end{equation}
\end{definition}

Thus, $\text{VaR}_{\alpha}(Z)$ is the lowest amount $z$ such that, with probability $\alpha$, $Z$ will not exceed $z$. In the case where $F$ is continuous, $\text{VaR}_{\alpha}(Z)$ is the unique $z$ satisfying $F(z) = \alpha$. Otherwise, it is possible that the equation $F(z) = \alpha$ has no solution, or an interval of solutions, depending upon the choice of $\alpha$. This, among other difficulties, motivates the use of alternative risk measures such as CVaR \cite{rockafellar2002conditional}.

Informally $\text{CVaR}_{\alpha}$ of $Z$ gives the expected value of $Z$, given that $Z\geq \text{VaR}_{\alpha}(Z)$. 
The precise definition is as follows. Let $[x]_+ = \max\{0,x\}$. 
\begin{definition}\label{CVaRdef}\cite{rockafellar2002conditional}
Let \begin{equation}\phi_{\alpha}(Z,\zeta) = \zeta + \frac{1}{1-\alpha}\mathbb{E}[Z-\zeta]_+.\end{equation}
Then $\text{CVaR}_{\alpha}(Z) = \min_{\zeta}\,\phi_{\alpha}(Z,\zeta)$, and \begin{equation}\text{VaR}_{\alpha}(Z)=\text{lower endpoint of }\underset{\zeta}{\arg\min}\,\phi_{\alpha}(Z,\zeta).\end{equation}
\end{definition}

It follows from the joint convexity of $\phi_{\alpha}$ in $Z$ and $\zeta$ that $\text{CVaR}_{\alpha}$ is convex over $\mathcal{Z}$. Restricting attention to random losses $Z(\omega) = L(x,\omega)$ which depend upon a decision $x$, we have the following result. 

\begin{theorem}\label{CVaRconvexinx} Let $Z(\omega) = L(x,\omega)$. If the mapping $x\mapsto Z$ is convex in $x$ then $\text{CVaR}_{\alpha}(Z)$ is convex in $x$ \cite{rockafellar2002conditional}.
\end{theorem}

Theorem \ref{CVaRconvexinx} will later ensure that optimization problems with objectives including a $\text{CVaR}_{\alpha}$ term are convex. 

\section{Risk Aware Stochastic Economic Dispatch Formulation}
We consider a setting with $N$ conventional generators, and a single renewable generator. An additional entity, the independent system operator (ISO) operates the power grid and plays the role of the social planner (from this point we use the terms interchangeably). For simplicity we consider a single bus network.

We consider a two-stage setting, where generation is dispatched in the first stage (also referred to as day-ahead or DA) and then adjusted in the second stage (real time or RT) to match demand. 

Let $D\geq0$ denote the aggregate demand. This demand is assumed \textit{inelastic}, i.e., it is not affected by changes in first or second stage prices. 

The renewable generator's output is modeled as a nonnegative random variable $W$, upper bounded by $\overline{W}>0$. We make the following additional assumption on the distribution of $W$. 

\begin{assumption}\label{Wassumption}
Random variable $W$ is distributed according to pdf $f_W$ (and cdf $F_W$), which is continuous and positive on $[0,\overline{W}]$. 
\end{assumption}

The probability distribution of $W$ is assumed to be known to all market participants. The marginal cost of renewable generation is zero. The quantity of renewable generation scheduled is denoted $y$. 

Conventional generator $i$ has access to a primary plant and an ancillary plant. Generator $i$ schedules its primary plant to produce quantity $x^G_i\in\mathbb{R}_+$ prior to realization of $W$ at cost $a_i(x^G_i)^2$ where $a_i>0$. We assume the primary plant is inflexible, so that its generation level must remained fixed once it is scheduled. After realization of $W$, generator $i$ can activate its ancillary plant to produce level $z^G_i\in\mathbb{R}_+$ at cost $\tilde{a}_i(z^G_i)^2$ where $\tilde{a}_i>0$. Any ancillary generation produced in excess of aggregate demand $D$ can be disposed of or sold in a separate spot market, which we do not consider. We assume that $a_i<\tilde{a}_i$ for all $i$, $a_i\neq a_j$ for $i\neq j$, and that $\max_ia_i<\min_i\tilde{a}_i$ and $\tilde{a}_i\neq \tilde{a}_j$ for $i\neq j$.

The generator is compensated for its first stage production $x^G_i$ at price $P_1$. In the second stage, given $W=w$, the generator is compensated for second stage generation $z^G_i(w)$ at price $P_2(w)$. 

\subsection{Generator's Problem}
We assume that each generator $i$ is \textit{price taking}, i.e., its decisions $x^G_i$ and $z^G_i(w)$ do not affect prices in either stage. Therefore, generator $i$'s profit is given by 
\begin{equation}\label{genprofit}\begin{split}\pi^G_i(x^G_{i},z^G_i(w)) := P_1&x^G_i - a_i(x^G_i)^2\\
 &+ P_2(w)z^G_i(w) - \tilde{a}_i(z^G_i(w))^2.
\end{split}
\end{equation}

Each generator is risk neutral, and so makes first and second stages to maximize the expectation of (\ref{genprofit}). In stage 2, given production level $w$ and price $P_2(w)$, generator $i$ solves the following problem 
\begin{equation}\label{genstage2}
(\text{GEN2}_i)\quad \max_{z^G_i(w)\geq 0}\quad P_2(w)z^G_i(w) - \tilde{a}_i(z^G_i(w))^2.
\end{equation}
Let $\pi^2_i(w,P_2)$ be the maximum objective value obtained in solving (\ref{genstage2}), given $w$ and $P_2$. Then in the first stage, given price $P_1$, generator $i$ solves the following problem 
\begin{equation}
(\text{GEN1}_i)\quad \max_{x^G_i\geq 0}\quad P_1x^G_i - c_i(x^G_i) + \mathbb{E}[\pi^2_i(W,P_2)].
\end{equation}
The term $\mathbb{E}[\pi^2_i(W,P_2)]$ is a constant when optimizing over $x^G_i$, as generator $i$'s DA and RT decisions can be made independently. In order to emphasize the fact that generator $i$ observes $W=w$ prior to selecting $z^G_i(w)$, we separate generator $i$'s two optimization problems. 
\subsection{ISO's Problem}
In Section III, our definition of a sequential competitive equilibrium includes a tuple of allocations, i.e., generation levels. For the purposes of examining the welfare properties of these allocations, we now introduce a two stage social planner's problem (SPP), corresponding to our two settlement market. As is seen in the static case, the SPP involves maximizing the social welfare of all market participants. We take the welfare of generator $i$ to be the negation of generation costs from stages 1 and 2. Given $W=w$, the aggregate welfare is the negation of the summation of these costs over all generators:

\begin{equation}\label{SPPwelfare}c^{\text{SPP}}(w):=\sum_i\left(a_i\hat{x}^2_i + \tilde{a}_i\hat{z}^2_i(w)\right),\end{equation} 
where $(\hat{x}_i,\hat{z}_i(w))$ for all $i$ and $w$ are the social planner's decisions in stages 1 and 2. Define $\hat{x}:=(\hat{x}_1,\dots,\hat{x}_N)$, and similarly $\hat{z}(w):=(\hat{z}_1(w),\dots,\hat{z}_N(w))$. 

We assume that the social planner is risk averse. That is, instead of seeking to minimize the expectation of (\ref{SPPwelfare}), they seek to minimize a weighted combination of $\mathbb{E}[c^{\text{SPP}}(W)]$ and $\text{CVaR}_{\alpha}(c^\text{SPP}(W))$. $\alpha\in[0,1)$ signifies that the ISO considers worst case or tail events with cumulative probability $1-\alpha$ to be ``risky'', and therefore weights them more heavily. We now introduce the additional parameter $\epsilon\in[0,1]$, which gives the social planner's relative weighting of overall expectation and $\text{CVaR}_{\alpha}$ of the first and second stage generation costs, and define the social planner's risk measure as 
\begin{equation}\rho_{\text{SPP}}(\cdot) = (1-\epsilon)\mathbb{E}[\cdot] + \epsilon\text{CVaR}_{\alpha}(\cdot).\end{equation}
It can be shown that $\rho_{\text{SPP}}(\cdot)$ is a coherent risk measure \cite{ShapDentRusz09}. 

Given that $\hat{y}$ is the amount of renewable generation scheduled by the social planner in stage 1, and $W=w$, the social planner's second stage problem is

\begin{align}
\label{secondstagestart}\text{(SPP2)}\quad\min_{\hat{z}(w)\geq 0}&\quad \sum_i\tilde{a}_i\hat{z}^2_i(w)\\
\label{recourseconstrSPP2}\text{s.t.}&\quad \sum_i\hat{z}_i(w) \geq \hat{y}-w.
\end{align}
Note that constraint (\ref{recourseconstrSPP2}) is an inequality in order to accommodate scenarios in which renewable generation exceeds residual demand $D-\sum_i\hat{x}_i=\hat{y}$. 

Define $c^{\text{SPP}}_2(\hat{x},w)$ as the minimum aggregate social cost achieved in the second stage, given $\hat{x}$ and $W=w$. Then the social planner's first stage problem is
\begin{align}
\text{(SPP1)}\quad\min_{\hat{x},\hat{y}\geq 0}&\quad \sum_ia_i\hat{x}^2_i + \rho_{\text{SPP}}\left(c^{\text{SPP}}_2(\hat{x},W)\right)\\
\text{s.t.}&\quad \sum_i\hat{x}_i + \hat{y} = D,
\end{align}
where we have used translation equivariance of $\text{CVaR}_{\alpha}$ to move the summed first-stage costs outside of $\rho_{\text{SPP}}$. 
We now argue that problems (SPP1) and (SPP2) can be combined into the following single stage optimization problem. 
\begin{lemma}\label{lem2stage1stage} The two-stage problem (SPP1)-(SPP2) is equivalent to the following single stage problem:
\begin{align}
\label{SPPstart}\text{(SPP)}\min_{\hat{x},\hat{y},\hat{z}(\cdot)\geq0}&\quad \sum_ic_i(\hat{x}_i) + \rho_{\text{SPP}}\left(\sum_i\tilde{a}_i\hat{z}^2_i(W)\right)\\
\label{firststageconstr}\text{s.t.}&\quad \sum_i\hat{x}_i + \hat{y} = D\\
\label{recourseconstr}&\quad \sum_i\hat{z}_i(w) \geq \hat{y}-w\quad\forall\,w,
\end{align}
where $\hat{z}(\cdot)\,:\,\mathbb{R}_+\to\mathbb{R}_+$. 
\end{lemma}
\begin{proof}
See Appendix. 
\end{proof}

Here we use the term ``equivalent'' in the sense that (SPP) and (SPP1) have the same optimal objective value. Additionally, if $(\hat{x}^*,\hat{z}^*(\cdot))$ is optimal for (SPP), then $\hat{x}^*$ is optimal for (SPP1), and $\hat{z}^*(w)$ is optimal for (SPP2) for all $w$, given $\hat{x}^*$. Conversely, if $\hat{x}^*$ is optimal for (SPP1) and $\hat{z}^*(\cdot)$ collects the optimal solutions to (SPP2) for all $w$, given $\hat{x}^*$, then $(\hat{x}^*,\hat{z}^*(\cdot))$ is optimal for (SPP). 

Similar to the equivalency demonstrated for the ISO's problem in Lemma \ref{lem2stage1stage}, it can be shown that the following single stage problem is equivalent to ($\text{GEN1}_i$) and ($\text{GEN2}_i$) 
\begin{equation}\label{gensinglestage}
\begin{split}
(\text{GEN}_i)\quad \max_{x^G_i\geq0,z^G_i(\cdot)\geq0}&\quad P_1x^G_i - a_i(x^G_i)^2\\
 &+ \mathbb{E}[P_2(w)z^G_i(w) - \tilde{a}_i(z^G_i(w))^2].
\end{split}
\end{equation}
where $z_i^G(\cdot)\,:\,\mathbb{R}_+\to\mathbb{R}_+$. 

 \section{Sequential Competitive Equilibrium}
 
 In a single stage market for a single good, a \textit{competitive equilibrium} is specified by a price $P$ and quantity $x$ such that, given $P$, producers find it optimal to produce, and consumers find it optimal to purchase, quantity $x$ of the good. Thus, the market clears, i.e., demand equals supply.

To understand the outcome of the two-stage market, we consider a sequential version of competitive equilibrium.

  \begin{definition}A \textit{sequential competitive equilibrium} (SCEq) is a tuple  $(\overline{x}^*,\overline{z}^*(\cdot),P^*_1,P^*_2(\cdot))$ such that, for all $i$, given $P^*_1$ and $P^*_2(\cdot)$, $\overline{x}^*_i$ is optimal for $(\text{GEN1}_i)$, $\overline{z}^*_i$ is optimal for $(\text{GEN2}_i)$, and there exists a $\overline{y}^*$, such that 
  \begin{equation}\sum_i\overline{x}^*_i + \overline{y}^* = D,\quad \sum_i\overline{z}^*_i(w) \geq \overline{y}^* - w\quad\forall\,w.\end{equation}
  \end{definition}  
  
  Note that in the SCEq definition, $\overline{z}_i^*(\cdot)$ and $P^*_2(\cdot)$ are functions. 
  We now investigate the existence of an SCEq in our two stage, risk aware setting. 

Let $\hat{\mu}(w)$ be the Lagrange multiplier corresponding to constraint (\ref{recourseconstrSPP2}). The Lagrangian for (SPP2) is

\begin{equation}\begin{split}\mathcal{L} &= \sum_i\tilde{a}_i\hat{z}^2_i(w)+\hat{\mu}(w)\left(\hat{y} - w-\sum_i\hat{z}_i(w)\right),\end{split}\end{equation}
giving, in addition to feasibility, the following optimality conditions for problem (SPP2):
\begin{align}
2\tilde{a}_i\hat{z}^*_i(w) - \hat{\mu}^*(w)&\geq 0\quad\forall\.i\\
\hat{z}_i^*\left(2\tilde{a}_i\hat{z}^*_i(w) - \hat{\mu}^*(w)\right)&= 0\quad\forall\.i\\
\hat{\mu}^*(w)\left(\hat{y}^* - w - \sum_i\hat{z}^*_i(w)\right)&=0\quad\forall\,w,\\
\hat{\mu}^*(w)&\geq 0\quad\forall\,w.
\end{align}
Assuming $\hat{y}>w$, $\hat{z}^*_i(w)>0$ for all $i$, and in particular
\begin{equation}\label{zstar} \hat{z}^*_i(w) = \frac{\hat{\mu}^*(w)}{2\tilde{a}_i}.\end{equation}
If $\hat{y}\leq w$ then $\hat{z}^*_i(w)=0$ for all $i$. Summing (\ref{zstar}) over $i$, applying constraint (\ref{recourseconstr}), and rearranging gives 
\begin{equation}
\hat{\mu}^*(w) =  2\tilde{a}\cdot\left[\hat{y}-w\right]_+,
\end{equation}
where the constant $\tilde{a}$ is defined as $\tilde{a}:=\left(\sum_{i}\frac{1}{\tilde{a}_i}\right)^{-1}$. Therefore 
\begin{equation}\nonumber \hat{z}^*_i(w) = \tilde{a}\cdot\frac{[\hat{y}-w]_+}{\tilde{a}_i}.\end{equation}
Summing over $i$ gives the optimal second stage objective value (i.e., the minimum recourse cost given $\hat{x}$)
\begin{equation}\label{mincost}c^{\text{SPP}}_2(\hat{x},w)= \sum_{i\in \mathcal{I}}\tilde{a}_i\left(\tilde{a}\cdot\frac{[\hat{y}-w]_+}{\tilde{a}_i}\right)^2 =\tilde{a}\cdot\left[\hat{y}-w\right]_+^2.\end{equation}

Therefore, $\text{VaR}_{\alpha}(c^{\text{SPP}}_2(\hat{x},W))$ may be expressed as
\begin{align}
\nonumber&\text{VaR}_{\alpha}\left(\tilde{a}\cdot[\hat{y}-W]_+^2\right) \\
&= \inf\left\{t\,:\,P\left(\tilde{a}\cdot\left[\hat{y}-W\right]_+^2\leq t\right)\geq \alpha\right\}\\
\nonumber&= \inf\left\{t\geq 0\,:\,P\left(W<\hat{y}-\sqrt{\frac{t}{\tilde{a}}}\right)\leq 1-\alpha\right\}\\
&= \begin{cases}
0&\text{if } \hat{y}<F^{-1}_W(1-\alpha)\\
\tilde{a}\cdot(\hat{y}-F^{-1}_W(1-\alpha))^2&\text{if }\hat{y}\geq F^{-1}_W(1-\alpha).
\end{cases}
\end{align}
Given this expression for $\text{VaR}_{\alpha}$, the following lemma gives an explicit expression of $\text{CVaR}_{\alpha}$ for our quadratic cost function setting. 
\begin{lemma}
Assuming first and second stage generation cost functions of the form $ax^2$ and $\tilde{a}z(w)^2$, $a,\tilde{a}>0$, $\text{CVaR}_{\alpha}(c^{\text{SPP}}_2(\hat{x},W))$ can be expressed as 
\begin{equation}\begin{split}
\label{CVaRintegral}&\text{CVaR}_{\alpha}\left(c^{\text{SPP}}_2(\hat{x},W)\right)=\text{CVaR}_{\alpha}\left(\tilde{a}\cdot[\hat{y}-W]_+^2\right)\\
&= \frac{1}{1-\alpha}\int_{0}^{\min\{F^{-1}_W(1-\alpha),\hat{y}\}}\tilde{a}\cdot(\hat{y}-w)^2\,f_W(w)\,dw.
\end{split}
\end{equation}
\end{lemma}
\begin{proof}
Given Assumption \ref{Wassumption}, the cdf $F_{c^{\text{SPP}}_2}$ of losses $c^{\text{SPP}}_2(\hat{y},W)$ will be continuous everywhere except possibly at zero, since $P(c^{\text{SPP}}_2(\hat{x},W)=0) = P(W \geq \hat{y})$. By Theorem 6.2 of \cite{ShapDentRusz09}, when $\text{VaR}_{\alpha}(c^{\text{SPP}}_2(\hat{x},W))>0$, we may write 
\begin{align}
\nonumber\text{CVaR}_{\alpha}&(c^{\text{SPP}}_2(\hat{x},W)) =\frac{1}{1-\alpha}\int_{\tilde{a}\cdot(\hat{y}-F^{-1}_W(1-\alpha))^2}^{\tilde{a}\hat{y}^2}qf_{c^{\text{SPP}}_2}(q)\,dq\\
\label{case1}&=\frac{1}{1-\alpha}\int_{0}^{F^{-1}_W(1-\alpha)}(\hat{y}-w)^2f_W(w)\,dw,
\end{align}
where $f_{c^{\text{SPP}}_2}$ gives the pdf corresponding to $F_{c^{\text{SPP}}_2}$. 
If $\text{VaR}_{\alpha}(c^{\text{SPP}}_2(\hat{x},W))=0$, then using Definition \ref{CVaRdef}, we have that 
\begin{align}\nonumber\text{CVaR}_{\alpha}&(c^{\text{SPP}}_2(\hat{x},W)) =  \zeta^* + \frac{1}{1-\alpha}\mathbb{E}[c^{\text{SPP}}_2(\hat{x},W)-\zeta^*]_+\\
\nonumber&=  0 + \frac{1}{1-\alpha}\mathbb{E}[c^{\text{SPP}}_2(\hat{x},W)- 0]_+\\
\label{case2}&=  \frac{1}{1-\alpha}\int_{0}^{\hat{y}} c^{\text{SPP}}_2(\hat{x},w)f_W(w)\,dw.
\end{align}
Substituting for $c^{\text{SPP}}_2(\hat{x},w)$ and then combining (\ref{case1}) and (\ref{case2}) completes the proof. 
\end{proof}

While $\text{CVaR}_{\alpha}(c^{\text{SPP}}_2(\hat{x},W))$ is convex in the first stage decision $\hat{x}$ due to (\ref{mincost}) and Theorem \ref{CVaRconvexinx}, the upper limit $\hat{\theta}$ of the integral in (\ref{CVaRintegral}) is not a differentiable function of $\hat{y}$, so that the Leibniz integral rule does not directly apply. The next lemma addresses this issue. 

\begin{lemma}\label{contdiff}
Given Assumption \ref{Wassumption}, expression (\ref{CVaRintegral}) is continuously differentiable with respect to $\hat{y}$, with derivative 
\begin{equation}\begin{split}\text{CVaR}'_{\alpha}(c^{\text{SPP}}_2(\hat{x},W)) &= 2\tilde{a}\int_{0}^{\hat{\theta}}(\hat{y}-w)f_W(w)\,dw.
\end{split}\end{equation}
\end{lemma}
\begin{proof}
We consider two cases, depending on the two possible values of $\hat{\theta}(\hat{y})$. When $\hat{\theta}(\hat{y}) = F^{-1}_W(1-\alpha)$, applying the Leibniz integral rule gives 
$$ \text{CVaR}'_{\alpha}(c^{\text{SPP}}_2(\hat{x},W)) = 2\tilde{a}\int_{0}^{F^{-1}_W(1-\alpha)}(\hat{y}-w)f_W(w)\,dw.$$
When $\hat{\theta}(\hat{y}) = \hat{y}$, application of the Leibniz integral rule gives 
$$ \text{CVaR}'_{\alpha}(c^{\text{SPP}}_2(\hat{x},W)) = 2\tilde{a}\int_{0}^{\hat{y}}(\hat{y}-w)f_W(w)\,dw.$$
Combining the last two equations gives the expression in the lemma statement. When $\hat{y}\leq F^{-1}_W(1-\alpha)$, $\text{CVaR}'_{\alpha}(c^{\text{SPP}}_2(\hat{x},W))$ is an affine function of $\hat{y}$, and when $\hat{y}> F^{-1}_W(1-\alpha)$, $\text{CVaR}'_{\alpha}(c^{\text{SPP}}_2(\hat{x},W))$ is continuous by the continuity of $f_W(w)$. The two expressions agree at $\hat{y}=F^{-1}_W(1-\alpha)$, so that $\text{CVaR}'_{\alpha}(c^{\text{SPP}}_2(\hat{x},W))$ is continuous. 
\end{proof}

Let $\hat{\theta} = \hat{\theta}(\hat{y}) = \min\{F^{-1}_W(1-\alpha),\hat{y}\}$. Then, problem (SPP) may be written as
\begin{align}\label{central}\min_{\hat{x},\hat{y},\hat{z}(\cdot)\geq 0}&\quad \sum_{i}a_i\hat{x}^2_i +(1-\epsilon)\int_{0}^{\hat{y}}\sum_i\tilde{a}_i\hat{z}^2_i(w)\,f_W(w)\,dw \\
\nonumber&\hspace{0.5in} + \frac{\epsilon}{1-\alpha}\int_{0}^{\hat{\theta}}\sum_i\tilde{a}_i\hat{z}^2_i(w)\,f_W(w)\,dw\\
\label{firststageconstr_rewrite}\text{s.t.}&\quad \sum_i\hat{x}_i + \hat{y} = D\\
\label{recourseconstr_rewrite}&\quad \sum_i\hat{z}_i(w) \geq \hat{y} - w\quad\forall\,w.
\end{align}

  Locational marginal pricing (LMP) is a commonly used settlement scheme for economic dispatch problems, and previous work has examined extensions of LMPs to problems including two stage markets with recourse. In such models, the LMPs arise as the dual variables to power balance constraints for each stage (in our setting (\ref{firststageconstr_rewrite}) and (\ref{recourseconstr_rewrite}) in (SPP)). Previous work (\cite{dahlin2019twostage},\cite{wannegkowshameysha11b}) has demonstrated that such LMPs support a competitive equilibrium when the ISO or social planner is risk neutral, i.e. when $\epsilon=0$. We state this formally in terms of our setting in the following theorem. 
  
  Let $\hat{\lambda}^*$ and $\hat{\mu}^*(w)$ denote the optimal Lagrange multipliers for constraints (\ref{firststageconstr_rewrite}) and (\ref{recourseconstr_rewrite}), given $W=w$, respectively. 
  
     \begin{theorem} When $\epsilon=0$, there exists an SCEq. In particular, $(\overline{x}^*,\overline{z}^*)$ are given by $(\hat{x}^*,\hat{z}^*)$ in the optimal solution to (SPP), and the equilibrium prices are given by 
\begin{equation}
              P^*_2(w) =  \hat{\mu}^*(w),\quad P^*_1  = \hat{\lambda}^*.
\end{equation}
  \end{theorem}
    \begin{proof} Our setting with $\epsilon=0$ can be seen as a special case of that in \cite{dahlin2019twostage}. The proof then follows from Theorem 1 in \cite{dahlin2019twostage}.
  \end{proof}
   \begin{theorem}\label{compequilexist} If $0\leq \epsilon<1$, then there exists a competitive equilibrium. In particular, $(\overline{x}^*,\overline{z}^*)$ are given by $(\hat{x}^*,\hat{z}^*)$, the optimal solution to problem (SPP), 
   and the equilibrium prices are given by 
  \begin{equation}\label{equilprices}
              P^*_2(w) =  \begin{cases}
\frac{\hat{\mu}^*(w)}{\left(1-\epsilon+ \frac{\epsilon}{1-\alpha}\right)}& 0\leq w\leq \hat{\theta}^*\\
\frac{\hat{\mu}^*(w)}{(1-\epsilon)}& \hat{\theta}^*\leq w< \hat{y}^*\\
0& \hat{y}^*\leq w\end{cases},\quad P^*_1  = \hat{\lambda}^*,
  \end{equation}
  where 
  $\hat{\theta}^*= \min\{F^{-1}_W(1-\alpha),\hat{y}^*\}$.
  \end{theorem}
\begin{proof}
By Lemma \ref{contdiff}, the objective and all constraints in (SPP) are continuously differentiable. Problem (\ref{central})-(\ref{recourseconstr_rewrite}) has Lagrangian 
\begin{equation}\nonumber\begin{split}
\mathcal{L}&= \sum_{i}a_i\hat{x}^2_i +(1-\epsilon)\int_{0}^{\hat{y}}\sum_i\tilde{a}_i\hat{z}^2_i(w)\,f_W(w)\,dw \\
&+ \frac{\epsilon}{1-\alpha}\int_{0}^{\hat{\theta}}\sum_i\tilde{a}_i\hat{z}^2_i(w)\,f_W(w)\,dw\\
& + \hat{\lambda}\left(D-\sum_i\hat{x}_i -\hat{y}\right)\\
& + \int\hat{\mu}(w)\left(\hat{y} - w - \sum_i\hat{z}_i(w)\right)f_W(w)\,dw.
\end{split}
\end{equation}
Let 
\begin{equation}
\hat{c}_{\epsilon,\alpha}(w) = \begin{cases}
1-\epsilon + \frac{\epsilon}{1-\alpha}& 0\leq w\leq \hat{\theta}^*\\
1-\epsilon &\hat{\theta}^*<w< \hat{y}^*\\
0&\hat{y}^*\leq w
\end{cases}.
\end{equation}
Then, in addition to feasibility, the optimality conditions for (\ref{central})-(\ref{recourseconstr_rewrite}) are \cite{ShapDentRusz09}:
\begin{align}
\label{SPPKKT1}2a_i\hat{x}^*_i - \hat{\lambda}^* &\geq 0\quad\forall\,i\\
\label{SPPKKT2}\hat{x}^*_i\left(2a_i\hat{x}^*_i - \hat{\lambda}^*\right) &= 0\quad\forall\,i\\
\label{SPPKKT3}-\hat{\lambda}^* + \int\hat{\mu}^*(w)f_W(w)\,dw &\geq 0\\
\label{SPPKKT4}\hat{y}^*\left(-\lambda^* + \int\hat{\mu}^*(w)f_W(w)\,dw \right) &= 0\\
\label{SPPKKT5}2\tilde{a}_i\hat{c}^*_{\epsilon,\alpha}(w)\hat{z}^*_i(w) - \hat{\mu}^*(w)&\geq 0\quad\forall\,w\\
\label{SPPKKT6}\hat{z}^*_i(w)\left(2\tilde{a}_i\hat{c}^*_{\epsilon,\alpha}(w)\hat{z}^*_i(w) - \hat{\mu}^*(w)\right)&= 0\quad\forall\,w\\
\label{SPPKKT7}\hat{\mu}^*(w)\left(\hat{y}^* - w - \sum_i\hat{z}^*_i(w)\right)&=0\quad\forall\,w,\\
\label{SPPKKT8}\hat{\mu}^*(w)&\geq 0\quad\forall\,w.
\end{align}

In addition to feasibility, the optimality conditions for $(\text{GEN}_i)$ are 
\begin{align}
\label{GENKKT1}2\tilde{a}_ix^{G*}_i-P_1  &\geq 0\\
\label{GENKKT2}x^{G*}_i\left(2\tilde{a}_ix^{G*}_i-P_1 \right) &= 0.\\
\label{GENKKT3}2\tilde{a}_iz^{G*}_i(w)-P_2(w)  &\geq 0\quad\forall\,w\\
\label{GENKKT4}z^{G*}_i(w)\left(2\tilde{a}_iz^{G*}_i(w)-P_2(w)\right) &= 0\quad\forall\,w.
\end{align}
In view of optimality conditions (\ref{SPPKKT5}) and (\ref{SPPKKT6}), we choose the following price schedule
\begin{equation}\nonumber\begin{split}
P_2(w) = \begin{cases}
\frac{\hat{\mu}^*(w)}{\left((1-\epsilon)+ \frac{\epsilon}{1-\alpha}\right)}& 0\leq w\leq \hat{\theta}^*\\
\frac{\hat{\mu}^*(w)}{(1-\epsilon)}& \hat{\theta}^*\leq w< \hat{y}^*\\
0& \hat{y}^*\leq w
\end{cases},\quad P_1  &= \hat{\lambda}^*.
\end{split}
\end{equation}
Given these choices, for each $i$, the optimality conditions for ($\text{GEN}_i$) become 
\begin{align}
\label{GENKKTsub1}2a_ix^{G*}_i - \hat{\lambda}^* &\geq 0\quad\forall\,i\\
\label{GENKKTsub2}x^{G*}_i\left(2a_ix^{G*}_i - \hat{\lambda}^* \right) &= 0\quad\forall\,i\\
\label{GENKKTsub3}2\tilde{a}_i\hat{c}_{\epsilon,\alpha}(w)z^{G*}_i(w) - \hat{\mu}^*(w)&\geq 0\quad\forall\,w\\
\label{GENKKTsub4}z^{G*}_i(w)\left(\hat{c}_{\epsilon,\alpha}(w)z^{G*}_i(w) - \hat{\mu}^*(w)\right)&= 0\quad\forall\,w.
\end{align}
Choosing $x^{G*}_i = \hat{x}^*_i$ for all $i$ and $z^{G*}_i(w) = \hat{z}^*_i(w)$ for all $i$ and $w$, (\ref{GENKKTsub1}) and (\ref{GENKKTsub2}) become identical to (\ref{SPPKKT1}) and (\ref{SPPKKT2}), and (\ref{GENKKTsub3}) and (\ref{GENKKTsub4}) become identical to (\ref{SPPKKT5}) and (\ref{SPPKKT6}). 

Therefore $x^{G*}_i = \hat{x}^*_i$ for all $i$, and $z^{G*}_i(w) = \hat{z}^*_i(w)$ for all $i$ and $w$, and the selected prices, together with $(\hat{x}^*_i,\hat{z}^*_i(w))$ for all $i$ and $w$  constitute an SCEq, and we have shown by construction the existence of an SCEq.\end{proof}

Assuming $\hat{z}^*_i(w)>0$ for any $i$ (and therefore for all $i$), the second stage price given in (\ref{equilprices}) can be rewritten in terms of the social planner's primal decision variables and the level of renewable generation. Rearranging the term in parenthesis in (\ref{SPPKKT6}) gives 
\begin{equation}\label{muexp}\hat{z}^*_i(w) = \frac{\hat{\mu}^*(w)}{\hat{c}_{\epsilon,\alpha}(w)}\quad \forall\,w\leq \hat{y}^*.\end{equation}
Summing both sides of (\ref{muexp}) over $i$ and using constraint (\ref{recourseconstr}) gives 
\begin{equation}\label{P_2rewrite}\hat{y}^*-w = \frac{\hat{\mu}^*(w)}{2\tilde{a}\cdot\hat{c}_{\epsilon,\alpha}(w)}\implies \frac{\hat{\mu}^*(w)}{\hat{c}_{\epsilon,\alpha}(w)} = 2\tilde{a}(\hat{y}^*-w).\end{equation}
Thus when $0\leq \epsilon<1$, we have $$P^*_2(W)  = 2\tilde{a}\cdot\left[\hat{y}^*-W\right]_+.$$
Given that $\hat{x}^*_i>0$ for any $i$ (and therefore for all $i$), a similar calculation gives 
$$\hat{\lambda}^* =  P^*_1 = 2a(D-\hat{y}^*),$$
where $a = \left(\sum_i\frac{1}{a_i}\right)^{-1}$. 

We now address the case where $\epsilon=1$, as prices given in the statement of Theorem \ref{compequilexist} cannot be applied directly in the case where $\hat{\theta}^* < \hat{y}^*$. Consider a sequence $\{\epsilon(k)\}$, where $\lim_{k\to\infty}\epsilon(k) = 1$. Then, suppressing the dependence of $\hat{\mu}(w)$ on $\epsilon$, and taking the limit as $k\to\infty$ on both sides of (\ref{P_2rewrite}) gives 
\begin{equation}\begin{split}\lim_{k\to\infty}\frac{\hat{\mu}^*(w)}{\hat{c}^*_{\epsilon(k),\alpha}(w)} 
&=2\tilde{a}\cdot\left(\lim_{k\to\infty}\hat{y}^*(\epsilon(k))-w\right).
\end{split}\end{equation}
The limit $\lim_{k\to\infty}\hat{y}^*(\epsilon(k))$ exists, as (SPP) may be solved for the case where $\epsilon=1$, and the optimal solution is unique given our assumptions on the generator cost function form. 

Therefore, it still holds in the case where $\epsilon=1$ that $P^*_2(W) = 2\tilde{a}\cdot[\hat{y}^*-W]_+$, and in turn a competitive equilibrium is given by $(\hat{x}^*,\hat{z}^*(\cdot), P^*_1,P^*_2(\cdot))$, where $P^*_1 = \hat{\lambda}^*$ and $P^*_2(W) = 2\tilde{a}\cdot[\hat{y}^*-W]_+$. Finally we give the following lemma on continuity of the equilibrium prices in $\epsilon$. 

\begin{lemma}\label{pricecontinuity}
The equilibrium prices given in Theorem \ref{compequilexist} are continuous in $\epsilon\in[0,1]$. 
\end{lemma}
\begin{proof}
See Appendix.
\end{proof}



\section{Two-Stage Mechanism for Risk Aware Electricity Market with Renewable Generation}
In the proof of Theorem \ref{compequilexist}, it was shown that the SCEq prices arise as optimal dual solutions to (SPP). If we assume that the generators are not strategic, and that all participants know the distribution of $W$, then the following mechanism implements the SCEq: \begin{enumerate}
\item[(1)] Each generator $i$ submits cost function coefficients $a_i$ and $\tilde{a}_i$.
\item[(2)] The ISO solves (SPP), and announces stage 1 price $P^*_1$ and stage 2 price schedule $P^*_2(\cdot)$ as given by (\ref{equilprices}). 
\item[(3)] Generator $i$ solves $(\text{GEN1}_i)$ and receives $P^*_1\overline{x}^{G*}_i$. 
\item[(4)] At the start of stage 2, the renewable generation output $W=w$ is observed by the generators. Generator $i$ solves $(\text{GEN2}_i)$ and pays $P^*_2(w)\overline{z}^{G*}_i(w)$. 
\item[(5)] Generator $i$ produces $\overline{x}^{G*}_{i}+ \overline{z}^{G*}_{i}(w)$. 
\end{enumerate}

\section{Conclusion}

In this paper we consider a two-stage electricity market model with a single customer and multiple generators, taking into account the risk preferences of the customer while assuming that the generators are risk neutral. Our goal has been to determine if a sequential competitive equilibrium exists in such a market, given this discrepancy in risk attitude. We show that such an equilibrium does exist by formulating the risk aware stochastic economic dispatch market as a two-stage stochastic program, and solving this problem to determine equilibrium energy procurements and prices. The equilibrium prices directly reflect the social planner's risk attitude. Given these prices, we specify a market mechanism for implementation of the equilibrium, assuming that the generators are not strategic. In future work we will incorporate network topology, multiple consumers, and strategic behavior in both the generators and consumers and general convex cost functions. 

\bibliography{ACC2020}
\bibliographystyle{plain}

\appendices
\section{Proof of Lemma \ref{lem2stage1stage}}
Since $\text{CVaR}_{\alpha}$ is a coherent risk function, by Proposition 6.5 of \cite{ShapDentRusz09} it is continuous. Also, $\text{CVaR}_{\alpha}$ is clearly a proper, monotonic risk function. 

For every possible realization $W=w$ and first stage decision $\hat{x}$, there always exists a feasible solution to second stage problem (\ref{secondstagestart})-(\ref{recourseconstr}), so $c^{\text{SPP}}_2(\hat{x},W)$ is finite with probability 1 for all feasible $\hat{x}$. Together with the quadratic forms of the first and second stage cost functions, this implies that $c^{\text{SPP}}_2(\hat{x},W)\in\mathcal{L}_1(\Omega,\mathcal{F},P)$ for all feasible $\hat{x}$.

Therefore by Proposition 6.37 of \cite{ShapDentRusz09} we can write 
\begin{equation}\text{CVaR}_{\alpha}(c^{\text{SPP}}_2(\hat{x},W)) = \inf_{\hat{z}(\cdot)\in \mathcal{G}(x,\cdot)}\,\text{CVaR}_{\alpha}\left(\sum_i\tilde{a}_i(\hat{z}_i(\cdot))\right),\end{equation}
where $\hat{z}(\cdot)\,:\,\Omega\to\mathbb{R}^N$. $\hat{z}(\cdot)\in\mathcal{G}(x,\cdot)$ denotes that $\hat{z}(w)$ is a feasible choice for the constraint set in problem (\ref{secondstagestart})-(\ref{recourseconstr}), given first stage decision $(\hat{x},\hat{y})$, for any $w$. 

The argument for interchanging expectation and minimization over $\hat{z}(\cdot)$ follows from the interchangeability principle (Theorem 7.80 in \cite{ShapDentRusz09}).$\hfill\blacksquare$

\section{Proof of Lemma \ref{pricecontinuity}} Let $F(\epsilon)$ denote the feasible set of (SPP), given parameter $\epsilon\in[0,1]$. From \cite{still2018lectures}, the local compactness (LC) of $F$ at some $\overline{\epsilon}$ is satisfied if there exists a $\delta>0$ and compact set $C_0$ such that $$\bigcup_{\|\epsilon-\overline{\epsilon}\|\leq \delta}F(\epsilon)\subset C_0.$$ Observing (SPP) is equivalent to a problem with the same objective and constraints, with the additional constraints that $\sum_i\hat{x}_i \leq D$, $\hat{y}\leq D$ and $\sum_i\hat{z}_i(w)\leq M$ for a large enough finite $M$, and that the feasible set of (SPP) does not depend upon $\epsilon$, LC is satisfied for any $\epsilon\in[0,1]$. 

From \cite{still2018lectures} the constraint qualification (CQ) holds for $F(\epsilon)$ at some $(\overline{\hat{x}},\overline{\hat{z}}(\cdot),\overline{\epsilon})$ with $(\overline{\hat{x}},\overline{\hat{z}}(\cdot))\in F(\overline{\epsilon})$ if there is a sequence $\{\hat{x}_{\nu},\hat{z}_{\nu}(\cdot)\}_{\nu\in\mathbb{N}}\to (\overline{\hat{x}},\overline{\hat{z}}(\cdot))$ such that, given $(\hat{x}_{\nu},\hat{z}_{\nu}(\cdot))$, (\ref{recourseconstr}) is satisfied with strict inequality for all $\nu$. Clearly this condition holds for all $\epsilon\in[0,1]$. 

Define the optimal solution set of (SPP), given $\epsilon$, as $S(\epsilon)$. Then, since LC and CQ are satisfied for all $\epsilon\in[0,1]$ by Lemma 5.6 in \cite{still2018lectures}, $S(\epsilon)$ is \textit{outer semicontinuous}, meaning that for all sequences $\{\hat{x}_{\nu}^*,\hat{z}^*(\cdot)_{\nu}, \epsilon_{\nu}\}_{\nu\in\mathbb{N}}$, with $\epsilon_{\nu}\to\overline{\epsilon}$ and $(\hat{x}_{\nu}^*,\hat{z}^*(\cdot)_{\nu})\in S(\epsilon_{\nu})$, there exists an $(\overline{\hat{x}}^*,\overline{\hat{z}}^*(\cdot))\in S(\overline{\epsilon})$ such that $\|\hat{x}_{\nu}^*-\overline{\hat{x}}^*\|\to0$ and $\|\hat{z}^*(\cdot)_{\nu}-\overline{\hat{z}}^*(\cdot)\|\to 0$ for $\nu\to\infty$. 

Due to the strict convexity of the first and second stage cost functions, the objective of (SPP) is strictly convex, so that when an optimal solution $(\hat{x}^*,\hat{z}^*(\cdot))$ exists, it is unique. Therefore, outer semicontinuity of the optimal primal solutions in $\epsilon$ is equivalent to continuity in $\epsilon$. Since the equilibrium prices depend continuously on the primal solutions to (SPP), the prices themselves are continuous at any $\epsilon\in[0,1]$. $\hfill\blacksquare$ 

\end{document}